\documentclass{amsart}
\usepackage{amsmath, latexsym, amsthm, amssymb, dsfont, graphics,epsfig}
\usepackage{mathrsfs}

\theoremstyle{plain}
\newtheorem{theorem}{Theorem}[section]
\newtheorem{proposition}[theorem]{Proposition}
\newtheorem{lemma}[theorem]{Lemma}
\newtheorem{corollary}[theorem]{Corollary}

\theoremstyle{definition}
\newtheorem{notation}[theorem]{Notation}
\newtheorem{example}[theorem]{Example}

\theoremstyle{remark}
\newtheorem{remark}[theorem]{Remark}

\input xy
\xyoption{all}

\def\au{\mathcal{A}}

\title[The Rank and HNP of Submonoids of a Free Monoid]{The Rank and Hanna Neumann Property of Some Submonoids of a Free Monoid\footnote{Contributed talk titled ``On the rank of the intersection of two submonoids of a free monoid" at \emph{A$^3$: Abstract Algebra and Algorithms Conference, Eger, Hungary, August 14-17, 2011.}}}
\author[S. N. Singh, K. V. Krishna]{Shubh Narayan Singh,  K. V. Krishna}
\address{Department of Mathematics\\
Indian Institute of Technology Guwahati\\
Guwahati, India}
\email{\{shubh,kvk\}@iitg.ac.in}

\keywords{Finitely generated monoids, semi-flower automata, rank, Hanna Neumann property.}
\subjclass{68Q70, 68Q45, 20M35.}

\begin{document}

\maketitle

\begin{abstract}
This work aims at further investigations on the work of Giambruno and Restivo \cite{giam08} to find the rank of the intersection of two finitely generated submonoids of a free monoid. In this connection, we obtain the rank of a finitely generated submonoid of a free monoid that is accepted by semi-flower automaton with two bpi's. Further, when the product automaton of two deterministic semi-flower automata with a unique bpi is semi-flower with two bpi's, we obtain a sufficient condition on the product automaton in order to satisfy the Hanna Neumann property.
\end{abstract}

\section{Introduction}

In \cite{howson54}, Howson obtained an upper bound for the rank of intersection of two finitely generated subgroups of a free group in terms of the individual ranks of subgroups. Thus, it is known that the intersection of two finitely generated subgroups of a free group  is finitely generated. In 1956, Hanna Neumann proved that if $H$ and $K$ are finite rank subgroups of a free group, then
\[\widetilde{rk}(H \cap K) \le 2\widetilde{rk}(H)\widetilde{rk}(K),\] where $\widetilde{rk}(N) = \max(0, rk(N)-1)$ for a subgroup $N$ of rank $rk(N)$. This is an improvement on Howson's bound. Further, Neumann conjectured that\\

\hfill $\widetilde{rk}(H \cap K) \le \widetilde{rk}(H)\widetilde{rk}(K)$,\hfill ($\star$)\\

\noindent which is known as Hanna Neumann conjecture \cite{neu56}.

In contrast, it is not true that the intersection of two finitely generated submonoids of a free monoid is finitely generated. Since Tilson's work \cite{til72} in 1972, through the work of Giambruno and Restivo \cite{giam08} in 2008, there are several contributions in the literature on the topic. Using automata-theoretic approach, Giambruno and Restivo have investigated an upper bound for the rank of the intersection of two submonoids of special type in a free monoid. In fact, for the special case, they have proved the Hanna Neumann property. Two submonoids $H$ and $K$ are said to satisfy  \emph{Hanna Neumann property} (in short, HNP), if $H$ and $K$ satisfy the inequality ($\star$).

This work extends the work of Giambruno and Restivo \cite{giam08} to another special class of submonoids. Here, we find the rank of a finitely generated submonoid of a free monoid that is accepted by semi-flower automaton with two bpi's. Further, we obtain a condition to extend HNP for the submonoids of a free monoid which satisfy the following condition $C$.
\begin{center}
\begin{minipage}{10.6cm}
Two submonoids of a free monoid are said to satisfy the condition $C$,\break if they are accepted by deterministic semi-flower automata, each with a unique bpi and their product automaton is semi-flower with two bpi's.
\end{minipage}
\end{center}

Rest of the paper is organized as follows. In Section 2, we present some preliminary concepts and results that are used in this work. Section 3 is dedicated to present the main results of the paper. We conclude the paper in Section 4.

\section{Preliminaries}

In this section, we present some background material from  \cite{bers85,giam07,giam08}. We try to confine to the terminology and notations given there so that one may refer to \cite{bers85,giam07,giam08} for those notions that are not presented here, if any.

Let $A$ be a finite set called an \emph{alphabet} with its elements as \emph{letters}. The free monoid over $A$ is denoted by $A^*$ and $\varepsilon$ denotes the empty word -- the identity element of $A^*$. It is known that every submonoid of $A^*$ is generated by a unique minimal set of generators. Thus, the \emph{rank} of a submonoid $H$, denoted by $rk(H)$, of $A^*$ is defined as the cardinality of the minimal set of generators $X$ of $H$, i.e. $rk(H) = |X|$. Further, the \emph{reduced rank} of a submonoid $H$ of $A^*$ is defined as $\max(0, rk(H)-1)$ and it is denoted by $\widetilde{rk}(H)$.

An \emph{automaton} $\au$ over an alphabet $A$ is a quadruple $(Q, I, T, \mathcal{F})$, where $Q$ is a finite set called the set of \emph{states}, $I$ and $T$ are subsets of $Q$ called the sets of \emph{initial} and \emph{final} states, respectively,
and $\mathcal{F}\subseteq Q\times A\times Q$ called the set of \emph{transitions}. Clearly, by denoting the states as vertices/nodes and the transitions as labeled arcs, an automaton can be represented by a digraph in which initial and final states shall be distinguished appropriately.

A \emph{path} in $\au$ is a finite sequence of consecutive arcs in its digraph. For $q_i \in Q$ ($0\le i \le k$) and $a_j \in A$ ($1 \le j \le k$), let
\[q_0 \xrightarrow{a_1} q_1 \xrightarrow{a_2} q_2 \xrightarrow{a_3} \cdots \xrightarrow{a_{k-1}} q_{k-1} \xrightarrow{a_k} q_k\] be a path $P$ in an automaton $\au$ that is starting at $q_0$ and ending at $q_k$. In this case, we write $i(P) =q_0$ and $f(P) = q_k$.  The word $a_1\cdots a_k \in A^*$ is the \emph{label of the path} $P$. For each state $q \in Q$, the \emph{null path} is a path from $q$ to $q$ labeled by $\varepsilon$.

A path in $\au$ is called \emph{simple} if all the states on the path are distinct. A path that starts and ends at the same state is called as a \emph{cycle}, if it is not a null path. A cycle with all its intermediate states are distinct is called a \emph{simple cycle}. A cycle that starts and ends in a state $q$ is called simple in $q$, if no intermediate state is equal to $q$. Other notions related to paths, viz. subpath, prefix and suffix, can be interpreted with their literal meaning or one may refer to \cite{giam08}.

Let $\au$ be an automaton. The \emph{language accepted/recognized by $\au$}, denoted by $L(\au)$, is the set of words that are labels
of paths from an initial state to a final state. A state $q\in Q$ is \emph{accessible} (respectively, \emph{coaccessible})
if there is a path from an initial state to $q$ (respectively, a path from $q$ to a final state). An automaton is called
\emph{trim} if all the states of the automaton are accessible and coaccessible. An automaton $\au = (Q, I, T, \mathcal{F})$ is \emph{deterministic} if it has a unique initial state, i.e. $|I| = 1$, and there is at most one transition defined for a state and a letter.

An automaton is called a \emph{semi-flower automaton} if it is trim with a unique initial state that is equal to a unique final state such that all the cycles visit the unique initial-final state.

If an automaton $\au = (Q, I, T, \mathcal{F})$ is semi-flower, we denote the initial-final state by $1$. In which case, we simply write $\au = (Q, 1, 1, \mathcal{F})$.  Further, let us denote by $C_{\au}$ the set of cycles that are simple in $1$ and by $Y_{\au}$ the set of their labels.

Now, in the following we state the correspondence between semi-flower automata and finitely generated submonoids of a free monoid.

\begin{theorem}[\cite{giam08}]\label{thm1}
If $\au$ is a semi-flower automaton over an alphabet $A$, then $Y_{\au}$ is finite and $\au$ recognizes the submonoid generated by $Y_{\au}$ in $A^*$. Moreover, if $\au$ is deterministic, then $Y_{\au}$ is the minimal set of generators of the submonoid recognized by $\au$.
\end{theorem}

In addition to the above result, given a finitely generated submonoid $H$ of the free monoid $A^*$, one can easily construct a semi-flower automaton $\au$ such that $L(\au) = H$. Here, to construct $\au$, one may choose a initial-final state and connect a petal to the initial-final state that corresponds to each word of a (finite) generating set of $H$.

With this basic information, we now present the two results of Giambruno and Restivo which will be generalized/extended in the present paper.

\begin{theorem}[\cite{giam08}]\label{th2.2}
If $\au = (Q, 1, 1, \mathcal{F})$ is a semi-flower automaton with a unique bpi, then \[rk(L(\au)) \le |\mathcal{F}| - |Q| + 1.\]
Moreover, if $\au$ is deterministic, then \[rk(L(\au)) = |\mathcal{F}| - |Q| + 1.\]
\end{theorem}
Here, a state $q$ of an automaton is called a \emph{branch point going in}, in short \emph{bpi}, if the indegree of $q$ (i.e. the number of arcs coming into $q$) is at least 2.

\begin{theorem}[\cite{giam08}]\label{giamhnp}
If $H$ and $K$ are the submonoids accepted by deterministic semi-flower automata $\au_H$ and $\au_K$, respectively,
each with a unique bpi such that $\au_H  \times \au_K$ is a semi-flower automaton with a unique bpi, then \[\widetilde{rk}(H \cap K) \le \widetilde{rk}(H)\widetilde{rk}(K).\]
\end{theorem}
Here, for automata $\au = (Q, 1, 1, \mathcal{F})$ and $\au' = (Q', 1', 1', \mathcal{F}')$ both over an alphabet $A$,  $\au  \times \au'$ is the \emph{product automaton} $(Q \times Q', (1, 1'), (1, 1'), \widetilde{\mathcal{F}})$ over the alphabet $A$ such that
\[((p, p'), a, (q, q')) \in \widetilde{\mathcal{F}} \Longleftrightarrow (p, a, q) \in \mathcal{F} \; \mbox{ and } (p', a, q' ) \in \mathcal{F}'\] for all $p, q \in Q$, $p', q' \in Q'$ and $a \in A$.

Notice that if $\au$ and $\au'$ are deterministic then so is $\au \times \au'$. But if $\au$ and $\au'$ are trim, then $\au \times \au'$ need not be trim. However, by considering only those states which are accessible and coaccessible, we can make the product automaton $\au \times \au'$ trim. This process does not alter the language accepted by $\au \times \au'$. In fact, we have
\[L(\au \times \au') = L(\au) \cap L(\au').\] Hence, if we state a product automaton $\au \times \au'$ is semi-flower, we assume that the trim part of $\au \times \au'$, without any further explanation.

In the hypothesis of Theorem \ref{giamhnp}, if the product automaton has more than one bpi, then it is not true that $H$ and $K$ satisfy HNP. This has been shown through certain examples in \cite{giam07,giam08}.  In the present work, first we observe that HNP fails if the product automaton has two bpi's. We demonstrate this in Example \ref{2bpi-ctr}. Then we proceed to investigate on the conditions to achieve HNP in case the product automaton has two bpi's.

We would require the following supplementary results from \cite{giam08} in our main results. Instead of reworking the details, we simply state in the required form. In these results, let the automata be over an alphabet $A$ of cardinality $n$; and for an automaton $\au = (Q, I, T, \mathcal{F})$ and $i \ge 0$
\[BPO_i(\au) = \{q \in Q\; |\; \mbox{out degree of } q = i\}.\]

\begin{proposition}\label{result-1}
 If $\au = (Q,1,1,\mathcal{F})$ is a deterministic semi-flower automaton over $A$, then
\[|\mathcal{F}| - |Q| = \displaystyle\sum_{i = 2}^n |BPO_i(\au)|(i - 1).\]
\end{proposition}

\begin{proposition}\label{result-2} Let $\au_1$ and $\au_2$ be two deterministic automata over $A$. If
$c_i = |BPO_i(\au_1)|$ and $d_i = |BPO_i(\au_2)|$, for each $i = 1, \ldots, n$, then
\[|BPO_t(\au_1\times \au_2)|\leq\displaystyle\sum_{t\leq r,s\leq n}c_r d_s.\]
\end{proposition}

\begin{proposition}\label{result-3}
Let  $\langle c_1,\ldots,c_n \rangle$ and $\langle d_1,\ldots,d_n \rangle$ be two finite sequences of natural numbers; then
\[\sum_{t = 2}^n(t-1)\left(\sum_{t\leq r\leq n}c_r \sum_{t\leq s\leq n} d_s \right) \leq
\left(\sum_{i = 2}^n (i-1)c_i\right)\left(\sum_{j = 2}^n (j-1)d_j \right).\]
\end{proposition}

\section{Main Results}

In this section we present two results. First we obtain the rank of a finitely generated submonoid of a free monoid, if it is accepted by a semi-flower automaton with exactly two bpi's. This generalizes the result of Giambruno and Restivo for semi-flower automata with a unique bpi.  Then we proceed to
obtain HNP for the submonoids of a free monoid that satisfy the condition $C$.

We begin with introducing a concise notation for a semi-flower automaton in which only the initial-final state, bpi's and the respective paths between them will be represented along with their labels. We call this as \emph{\underline{b}pi's and \underline{p}aths \underline{r}epresentation}, in short \emph{BPR}, of an automaton. For example, the BPR of the semi-flower automaton given in \textsc{Figure} \ref{fig3} is shown in \textsc{Figure} \ref{fig4}.
\begin{figure}[htp]
\entrymodifiers={++[o][F-]} \SelectTips{cm}{}
\[\xymatrix{*\txt{} & *\txt{} & *\txt{} \ar[r] &  *++[o][F=]{1} \ar[lld]_a \ar[rrd]^b\\
*\txt{} & \ar[ld]_a \ar[rd]^b & *\txt{} & \ar[u]^a & *\txt{} &  \ar[ld]_a \ar[rd]^b \ar[dd]^a \\
\ar[rr]^b \ar[rd]_a & *\txt{} & q \ar[r]^a  & \ar[u]_b & \ar[rd]^b & *\txt{} & \ar[ld]^a \\
*\txt{} & \ar[ur]^b \ar[dr]_a & *\txt{} & \ar[ul]_a & \ar[l]_b & p \ar[l]_a \ar[ld]^b \\
*\txt{} & *\txt{}  & \ar[uu]^b & \ar[uul]^b & \ar[l]^a}\]
\caption{A Semi-Flower Automaton}
\label{fig3}
\end{figure}
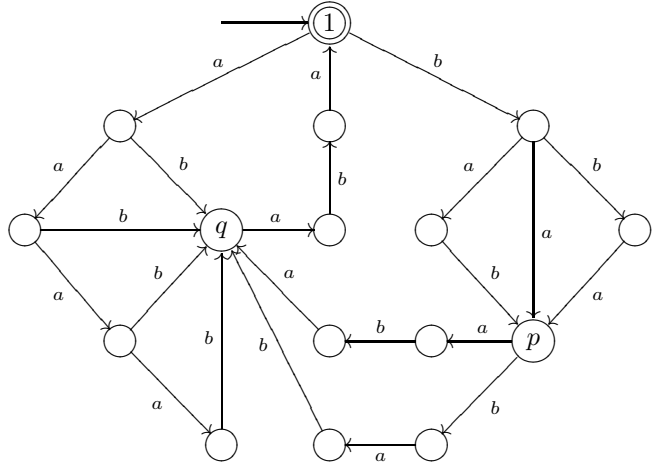

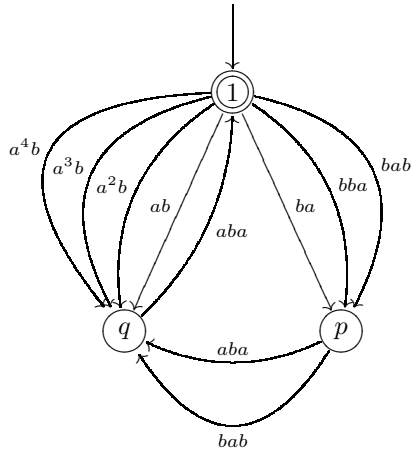
\begin{figure}[htp]
\entrymodifiers={++[o][F-]} \SelectTips{cm}{}
\[\xymatrix{  *\txt{} &  *\txt{} \ar[d]\\
 *\txt{} &  *++[o][F=]{1} \ar[dddr]^{ba} \ar@/^1.5pc/[dddr]^{bba} \ar@/^3pc/[dddr]^{bab} \ar[dddl]_{ab} \ar@/_1.5pc/[dddl]_{a^2b}
\ar@/_3pc/[dddl]_{a^3b} \ar@/_4.5pc/[dddl]_{a^4b} \\
*\txt{}\\
*\txt{}\\
q \ar@/_1pc/[uuur]_{aba} & *\txt{} & p \ar@/^1pc/[ll]_{aba} \ar@/^3pc/[ll]^{bab}}\]
\caption{BPR of the Semi-Flower Automaton given in \textsc{Figure} \ref{fig3}}
\label{fig4}
\end{figure}

The following lemma is useful for obtaining the rank of a semi-flower automaton with two bpi's.

\begin{lemma}
If $\au$ is a semi-flower automaton with exactly two bpi's, say $p$ and $q$ such that the distance from $q$ to the $1$
is not more than that of $p$, then
\begin{enumerate}
\item[\rm(i)] there is a unique simple path from $q$ to $1$, and
\item[\rm(ii)] every cycle in $\au$ visits $q$.
\end{enumerate}
\end{lemma}

\begin{proof}\
\begin{enumerate}
\item[(i)] If $q = 1$, then we are done. If not, by the choice of $q$, the initial-final state $1$ is not a bpi. Moreover, since $q$ is coaccessible, there is a path from $q$ to $1$. Now suppose there are two different paths $P_1$ and $P_2$ with labels $u$ and $v$, respectively, from $q$ to $1$. Let $w$ be the label of longest suffix path $P'$ which is in common between the paths $P_1$ and $P_2$. As $1$ is not a bpi, $w \neq \varepsilon$. But then $i(P')$ will be a bpi different from $q$. This a contradiction to the choice of $q$. Thus, there is a unique simple path from $q$ to $1$.

\item[(ii)] Since every cycle in $\au$ passes through $1$, if $q = 1$, then we are done. If not, $1$ is not a bpi. Now suppose there is a cycle that is not visiting $q$. Then the cycle contributes one to the indegree of the state $1$. Also, from above (i), there is a path from $q$ to $1$. This implies that the state $1$ is a bpi; a contradiction.
\end{enumerate}
\end{proof}

\begin{corollary}
Every cycle that visits $p$ also visits $q$ and hence, if $p$ and $q$ are distinguishable, the distance from $p$ to $1$ is more than that of $q$.
\end{corollary}

\begin{notation}\label{nota}
In what follows, if a semi-flower automaton has exactly two bpi's, say $p$ and $q$, then we consider that the distance from $q$ to $1$ is not
more than that of $p$. Moreover, we assume that the indegree of $p$ is $m$ and the indegree of $q$ is $(l + k)$, where $k$
is the number of edges ending at $q$ that are not in any of the paths from $p$ to $q$. With this information, the BPR of such an automaton will be as shown in \textsc{Figure} \ref{fig5}.
\begin{figure}[htp]
\entrymodifiers={++[o][F-]} \SelectTips{cm}{}
\[\xymatrix{  *\txt{} &  *\txt{} \ar[d]\\
 *\txt{} &  *++[o][F=]{1} \ar[dddr]^{\cdots m \cdots} \ar@{..>}@/^1pc/[dddr] \ar@{..>}@/^2pc/[dddr] \ar[dddl]_{\cdots k \cdots} \ar@{..>}@/_1pc/[dddl] \ar@{..>}@/_2pc/[dddl] \\
*\txt{}\\
*\txt{}\\
q \ar@/_1pc/[uuur] & *\txt{} & p \ar@/_1pc/[ll]^{\stackrel{\vdots}{\stackrel{l}{\vdots}}} \ar@{..>}@/^1.5pc/[ll] \ar@{..>}@/^3pc/[ll]}\]
\caption{BPR of a Semi-Flower Automaton with Two bpi's}
\label{fig5}
\end{figure}
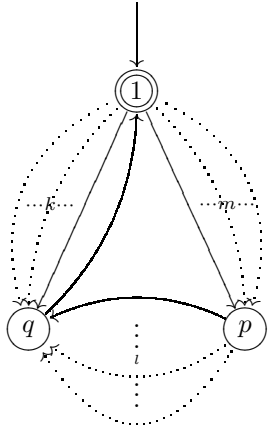
\end{notation}

Now we are ready to present our first result of the paper.

\begin{theorem}\label{2brk}
If $\au$ is a semi-flower automaton with exactly two bpi's $p$ and $q$, then \[rk(L(\au)) \le ml + k.\] Moreover, if $\au$ is deterministic, then
\[rk(L(\au)) = ml + k.\]
\end{theorem}

\begin{proof} As the number of simple cycles passing through the initial-final state $1$ (i.e. in $C_\au$) gives us an upper bound for the rank $rk(L(\au))$, we count these cycles using indegree of $p$ and $q$.
The number of cycles in $C_\au$ that are passing through $q$ but not $p$ is $k$. Also, as each path entering the state $p$ will split into $l$ number of paths and enter in the state $q$, we have $ml$ number of cycles in $C_\au$ that are passing through $p$. Thus, the total number of cycles in $C_\au$ is $ml + k$. Hence, as $L(\au)$ is the submonoid generated by $Y_\au$, we have
\[rk(L(\au)) \le |Y_{\au}| = |C_\au| = ml + k.\]
If $\au$ is deterministic, then by Theorem \ref{thm1}, we have \[rk(L(\au)) = ml + k.\]
\end{proof}

In a semi-flower automaton with two indistinguishable bpi's, i.e. with a unique bpi,  we have the following corollary.

\begin{corollary}
Theorem \ref{th2.2} follows.
\end{corollary}

\begin{proof}
In the hypothesis of Theorem \ref{2brk}, if $p = q$  (i.e. $p$ and $q$ are indistinguishable), then $\au$ has a unique bpi. In which case, $l = 0$ and consequently, $rk(L(\au))$ is less than or equal to  the indegree $k$ of the unique bpi. And in case $\au$ is deterministic, $rk(L(\au)) = k$.  Now the number of transitions $|\mathcal{F}|$ in $\au$ can be counted by the number of arcs entering all the states of $\au$. As $\au$ is trim, every state of $\au$ has an arc into it. Further, since $\au$ has a unique bpi, except the bpi, all other states have indegree one. Thus, we have
\[|\mathcal{F}| = |Q| - 1 + k,\] so that \[rk(L(\au)) \le |\mathcal{F}| - |Q| + 1.\] Moreover, if $\au$ is deterministic, then the equality holds.
\end{proof}

\begin{remark}
Theorem \ref{2brk} generalizes Theorem \ref{th2.2}.
\end{remark}

Before proceeding to our second result, it is appropriate to note the following example.

\begin{example}\label{2bpi-ctr}
Consider the submonoids $H = \{aa, aba, ba, bb\}^*$ and $K = \{a, bab\}^*$ of the free monoid $\{a, b\}^*$. We give the automata $\au_H$ and $\au_K$ which accept $H$ and $K$, respectively, in \textsc{Figure} \ref{fig1}. Note that $\au_H$ and $\au_K$ are deterministic semi-flower automata, each with unique bpi. The (trim form of) product automaton $\au_H \times \au_K$ is shown in \textsc{Figure} \ref{fig2}. Clearly, $\au_H \times \au_K$ is semi-flower with exactly two bpi's, viz. $(1, 1')$ and $(1, 3')$ and hence $rk(H \cap K) = 5$. Whereas, $rk(H) = 4$ and $rk(K) = 2$. Thus, $H$ and $K$ do not satisfy HNP, i.e. \[\widetilde{rk}(H\cap K) > \widetilde{rk}(H)\widetilde{rk}(K).\]
\end{example}

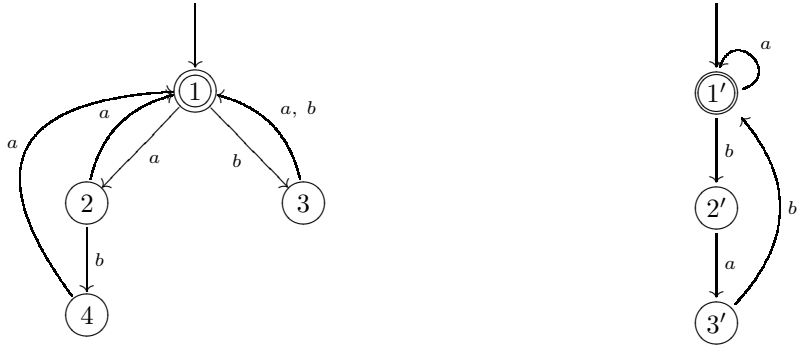
\begin{figure}[htb]
\entrymodifiers={++[o][F-]} \SelectTips{cm}{}
\[\xymatrix{
*\txt{} &  *\txt{}\ar[d]\\
*\txt{} & *++[o][F=]{1}\ar[dl]^a \ar[dr]_b\\
2 \ar@/^1pc/[ur]^{a} \ar[d]^b & *\txt{} & 3 \ar@/_1pc/[ul]_{a, \ b}\\
4 \ar@/^4pc/[uur]^a}
\hspace{4cm}
\xymatrix{
*\txt{} &  *\txt{}\ar[d]\\
*\txt{} & *++[o][F=]{1'}\ar[d]^b  \ar@(r,u)[]_a\\
*\txt{} & 2' \ar[d]^a\\
*\txt{} & 3' \ar@/_2pc/[uu]_b}\]
\caption{$\au_H$ (in the left) and $\au_K$ (in the right) of Example \ref{2bpi-ctr}}
\label{fig1}
\end{figure}

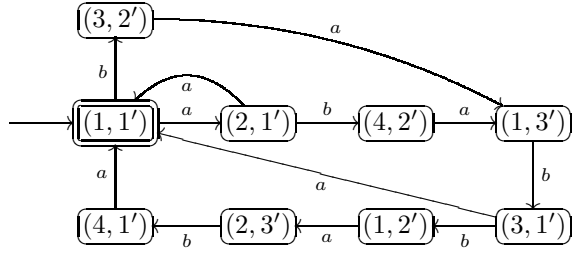
\begin{figure}[htp]
\entrymodifiers={+<3pt>[F-:<3pt>]} \SelectTips{cm}{}
\[\xymatrix{
*\txt{} & (3,2')\ar@/^1pc/[drrr]^a \\
*\txt{} \ar[r] & *+<7pt>[F=:<3pt>]{(1,1')} \ar[u]^b \ar[r]^a & (2,1') \ar@/_1.5pc/[l]^a \ar[r]^b
 & (4,2') \ar[r]^a & (1,3') \ar[d]^b \\
*\txt{} & (4,1') \ar[u]^a & (2,3') \ar[l]^b & (1,2') \ar[l]^a & (3,1') \ar[ulll]^a \ar[l]^b }\]
\caption{$\au_H \times \au_K$ of Example \ref{2bpi-ctr}}
\label{fig2}
\end{figure}

The following lemma is useful in proving our second result of the paper.

\begin{lemma} \label{bpopr}
If $\au = (Q,1,1,\mathcal{F})$ is a semi-flower automaton with exactly two bpi's $p$ and $q$,
then \[rk(L(\au)) -(m - 1)(l - 1) \le |\mathcal{F}|- |Q| + 1.\]
Moreover, if $\au$ is deterministic, then the equality holds.
\end{lemma}

\begin{proof}
Since the number of transitions $|\mathcal{F}|$ of $\au$ is the total indegree (i.e. the sum of indegrees of all the states) of the digraph of $\au$, we have
\[|\mathcal{F}| = m + l + k + |Q| - 2.\] Consequently,
\begin{eqnarray*}
|\mathcal{F}|- |Q| + 1 &=& m + l + k - 1\\
\Longrightarrow \ \  |\mathcal{F}|- |Q| + 1 &=& (ml + k) -(ml - m - l + 1)\\
\Longrightarrow \ \ |\mathcal{F}|- |Q| + 1 &=& (ml + k) -(m - 1)(l - 1).
\end{eqnarray*}
Hence, by Theorem \ref{2brk}, $|\mathcal{F}|- |Q| + 1 \ge rk(L(\au)) -(m - 1)(l - 1).$
\end{proof}

Now, by Theorem \ref{result-1}, we have the following corollary.
\begin{corollary}\label{cor3.4}
If $\au$ is a deterministic semi-flower automaton with exactly two bpi's $p$ and $q$,
then $rk(L(\au))  = (m - 1)(l - 1) + \displaystyle\sum_{t = 2}^n |BPO_t(\au)|(t - 1) + 1.$
\end{corollary}

\begin{theorem}\label{2bhn}
If $H$ and $K$ are the submonoids accepted by deterministic semi-flower automata $\au_H$ and $\au_K$, respectively, each with a unique bpi
such that the product automaton $\au_H  \times \au_K$ is semi-flower with exactly two bpi's $p$ and $q$, then
\[\widetilde{rk}(H \cap K) \le \widetilde{rk}(H)\widetilde{rk}(K) + (m-1)(l-1).\]
\end{theorem}

\begin{proof} Note that
\begin{eqnarray*}
\widetilde{rk}(H\cap K) &=& rk(L(\au_H  \times \au_K))  - 1\\
 &=& (m - 1)(l - 1) + \displaystyle\sum_{t = 2}^n |BPO_t(\au_H  \times \au_K)|(t - 1) \ \ \mbox{by Corollary \ref{cor3.4}}\\
 &\leq& (m - 1)(l - 1) + \sum_{t = 2}^n (t - 1)\left(\sum_{t\leq r, s\leq n}c_rd_s \right)  \ \ \mbox{by Proposition \ref{result-2}},
\end{eqnarray*}
where $c_r = |BPO_r(\au_H)|$ and $d_s = |BPO_s(\au_K)|$. Consequently, by Proposition \ref{result-3}
\begin{eqnarray*}
\widetilde{rk}(H\cap K) &\leq& (m - 1)(l - 1) + \left(\sum_{i = 2}^n (i - 1)c_i\right)\left(\sum_{j = 2}^n (j-1)d_j\right)\\
 &=& (m - 1)(l - 1) + \widetilde{rk}(H)\widetilde{rk}(K) \ \ \mbox{by Theorem \ref{th2.2} and Proposition \ref{result-1}. }
\end{eqnarray*}
Hence the result.
\end{proof}

\begin{corollary}
In addition to the hypothesis of Theorem \ref{2bhn}, if there is a unique path from $p$ to $q$ in $\au_H  \times \au_K$, then
\[\widetilde{rk}(H \cap K) \le \widetilde{rk}(H)\widetilde{rk}(K).\]
\end{corollary}

\section{Conclusion}

In this work we have obtained the rank of a finitely generated submonoid of a free monoid that is accepted by a semi-flower automaton with two bpi's. This generalizes the rank result (cf. Theorem \ref{th2.2}) for semi-flower automata with unique bpi by Giambruno and Restivo \cite{giam08}. In fact, the present proof of Theorem \ref{th2.2} is shorter and elegant than that of the original proof by  Giambruno and Restivo. In \cite{giam08}, Giambruno and Restivo obtained HNP for submonoids of a free monoid that are accepted by deterministic semi-flower automata, each with a unique bpi such that their product automaton is semi-flower with a unique bpi. Further, by keeping the former automata as they are, if the latter automaton has more than one bpi, they provided examples which fail to satisfy HNP. In the present work, we give an example which fails to satisfy HNP when the product automaton has two bpi's. In case the product automaton has exactly two bpi's,  we reported a sufficient condition to obtain HNP. The techniques introduced in this work shall give a scope to one in extending our work to a general scenario.

\end{document}